\documentclass[submission,copyright,creativecommons]{eptcs}
\providecommand{\event}{ICLP 2020} 
\usepackage{breakurl}              
\usepackage[english]{babel}

\makeatletter
\let\O@argtabularcr\@argtabularcr
\def\O@xtabularcr{\@ifnextchar[\O@argtabularcr{\ifnum 0=`{\fi}\cr}}
\let\O@tabacol\@tabacol
\let\O@tabclassiv\@tabclassiv
\let\O@tabclassz\@tabclassz
\let\O@tabarray\@tabarray
\def\author@tabular{\authorsize\def\@halignto{}\@authortable}
\let\endauthor@tabular=\endtabular
\def\author@tabcrone{{\ifnum0=`}\fi\O@xtabularcr\affilsize\itshape
 \let\\=\author@tabcrtwo\ignorespaces}
\def\author@tabcrtwo{{\ifnum0=`}\fi\O@xtabularcr[-3\p@]\affilsize\itshape
 \let\\=\author@tabcrtwo\ignorespaces}
\def\@authortable{\leavevmode \hbox \bgroup $\let\@acol\O@tabacol
 \let\@classz\O@tabclassz \let\@classiv\O@tabclassiv
 \let\\=\author@tabcrone \ignorespaces \O@tabarray}
\makeatother

\usepackage{array}
\usepackage[table]{xcolor}    
\usepackage{hhline}
\usepackage{mathptmx}
\usepackage{latexsym}
\usepackage{amsfonts,amssymb,latexsym}

\usepackage{amsmath,amsthm}
\usepackage{url}
\usepackage{multirow}
\usepackage{color}
\usepackage{adjustbox}
\usepackage{xypic}
\input xy
\xyoption{all}

\newcommand{\?}{\stackrel{?}{=}}

\def\cC{\mathcal{C}}
\def\cD{\mathcal{D}}

\long\def\comment#1{}

\newcommand{\caE}{\ensuremath{\mathcal E}}

\newcommand{\caT}{\ensuremath{\mathcal T}}

\newcommand{\caX}{\ensuremath{\mathcal X}}

\newcommand{\sort}[1]{\ensuremath{\mathsf{#1}}}

\newcommand{\Pos}{{\mathit P}os}
\newcommand{\Variables}{\caX}
\newcommand{\Symbols}{\Sigma}
\newcommand{\DSymbols}{\cD}
\newcommand{\CSymbols}{\cC}

\newcommand{\TermsOn}[5]{{\caT^{#4}_{#1}(#2)}_{#3}^{#5}}
\newcommand{\Terms}{\TermsOn{\Symbols}{\Variables}{}{}{}}
\newcommand{\TermsS}[1]{\TermsOn{\Symbols}{\Variables}{\sort{#1}}{}{}}

\newcommand{\GTermsOn}[2]{\caT^{#2}_{#1}}
\newcommand{\GTerms}{\GTermsOn{\Symbols}{}}
\newcommand{\GTermsS}[1]{\GTermsOn{\Symbols,\sort{#1}}{}}

\newcommand{\SubstOn}[2]{{\cal S}ubst(#1,#2)}
\newcommand{\Substs}{\SubstOn{\Symbols}{\Variables}{}{}{}}
\newcommand{\idsubst}{\textit{id}}

\newcommand{\composeSubst}{}
\newcommand{\composeRel}{;}
\newcommand{\compose}{\composeSubst}
\newcommand{\restrict}[2]{#1|_{#2}}

\newcommand{\congr}[1]{=_{#1}}

\newcommand{\csuV}[3]{\textit{CSU\/}^{#2}_{#3}({#1})}

\newcommand{\var}[1]{\mathit{Var}(#1)}
\newcommand{\occ}[1]{\mathit{Pos}(#1)}
\newcommand{\occSub}[2]{\mathit{Pos}_{#2}(#1)}
\newcommand{\funocc}[1]{\mathit{Pos}_{\Symbols}(#1)}

\newcommand{\subterm}[2]{#1|_{#2}}
\newcommand{\replace}[3]{#1[#3]_{#2}}
\newcommand{\domain}[1]{\mathit{Dom}(#1)}
\newcommand{\range}[1]{\intrvar{#1}}
\newcommand{\intrvar}[1]{\mathit{Ran}(#1)}

\newcommand{\rootpos}{\Lambda}

\newcommand{\rewrite}[1]{\rightarrow_{#1}}
\newcommand{\rewrites}[1]{\rightarrow^*_{#1}}

\newcommand{\narrow}[2]{\mathop{\stackrel{#1}{\rightsquigarrow}_{#2}}}

\newcommand{\norm}[1]{{\downarrow_ {#1}}}

\newcommand{\sem}[1]{{[\![#1]\!]}_{E,B}}

\usepackage{ifthen}

\newenvironment{flemma-noname}[2][]{\vskip\topsep\noindent{\bf
Lemma #2\ifthenelse{\equal{#1}{}}{}{\ }#1.}\em\ }{\vskip\topsep}

\newcommand{\ignore}[1]{}

\newtheorem{example}{Example}
\newtheorem{definition}{Definition}
\newtheorem{lemma}[definition]{Lemma}
\newtheorem{proposition}[definition]{Proposition}
\newtheorem{corollary}[definition]{Corollary}

\title{Variant-based Equational Unification under\\ Constructor Symbols\thanks{This work has been partially supported by the EU (FEDER) and the Spanish MCIU under grant RTI2018-094403-B-C32, by the Spanish Generalitat Valenciana under grants PROMETEO/2019/098 and APOSTD/2019/127, and by the US Air Force Office of Scientific Research under award number FA9550-17-1-0286.
}}
\author{
Dami\'an Aparicio-S\'anchez \qquad Santiago Escobar \qquad Julia Sapi\~{n}a
\institute{VRAIN (Valencian Research Institute for Artificial Intelligence)\\Universitat Polit\`ecnica de Val\`encia\\
Valencia, Spain}
\email{\{daapsnc,sescobar,jsapina\}@upv.es}
}

\def\titlerunning{Variant-based Equational Unification under Constructor Symbols}
\def\authorrunning{D. Aparicio-S\'anchez, S. Escobar, and J. Sapi\~{n}a}
\begin{document}
\maketitle

\begin{abstract}
Equational unification of two terms consists of finding a substitution that, when applied to both terms, makes them equal modulo some equational properties. A narrowing-based equational unification algorithm relying on the concept of the \emph{variants} of a term is available in the most recent version of Maude, version 3.0, which provides quite sophisticated unification features. A variant of a term t is a pair consisting of a substitution $\sigma$ and the canonical form of $t\sigma$. Variant-based unification is decidable when the equational theory satisfies the \emph{finite variant property}. However, this unification procedure does not take into account constructor symbols and, thus, may compute many more unifiers than the necessary or may not be able to stop immediately. In this paper, we integrate the notion of constructor symbol into the variant-based unification algorithm. Our experiments on  positive and negative unification problems show an impressive speedup.
\end{abstract}

\section{Introduction}\label{sec:intro}

Equational unification of two terms is of special relevance to many areas in computer science, including logic programming, and consists of finding a substitution that, when applied to both terms, makes them equal modulo some equational properties. Several algorithms have been developed in the literature for specific equational theories, such as associative-commutative symbols, exclusive-or, Diffie-Hellman, or Abelian Groups (see~\cite{BS-handbook00}). Narrowing was proved to be complete for unification \cite{JKK83} and several cases have been studied where narrowing provides a decidable unification algorithm~\cite{AEI09,AEI11}. A narrowing-based equational unification algorithm relying on the concept of the \emph{variants} of a term \cite{CD05} has been developed in \cite{ESM12} and it is available in the most recent version of Maude, version 3.0, which provides quite sophisticated unification features \cite{maude-manual,DEEM+20}.

Several tools and techniques rely on Maude's advanced unification capabilities, such as termination~\cite{DLM09} and local confluence 
and coherence~\cite{DM12,DMR20} proofs, 
narrowing-based theorem proving~\cite{Rusu10} or testing~\cite{Riesco14}, and \emph{logical model checking} \cite{EM07,BEM13}. The area of cryptographic protocol analysis has also benefited from advanced unification algorithms: 
Maude-NPA~\cite{EMM09},
Tamarin~\cite{DHRS18} 
and AKISS~\cite{BDGK17}
rely on the different unification features of Maude.
Furthermore, numerous decision procedures for formula satisfiability modulo equational theories also rely on unification, either based on narrowing \cite{TGRK15} or by using variant generation in finite variant theories~\cite{Meseguer18-scp}.

Constructor symbols are extensively used in computer science: for representing data instead of functions, for manipulating programs as data, or for reasoning in complex semantic structures. 
In an equational theory, constructors can be characterized in the ``no junk, no confusion" style of Goguen and Burstall \cite{BG82},
providing the mathematical semantics of the equational theory as the \emph{initial algebra} of a Maude functional module, 
which corresponds to the least Herbrand model in logic programming
(see \cite{maude-manual}).
However, this more general notion of constructor differs from 
the ``logic" notion of a functor 
and 
the ``functional" 
notion of a symbol not appearing in the root position of the left-hand side of any equation. 
The notion of a constructor symbol has not yet been integrated into
the variant-based equational unification procedure of Maude and, thus, it may compute many more unifiers than the necessary or it may not be able to stop immediately. In this paper, we integrate the notion of constructor symbol into the variant-based unification algorithm
with an impressive speedup.

After some preliminaries in Section~\ref{sec:preliminaries}, we recall variant-based unification in Section~\ref{sec:eq_unification}. 
In Section~\ref{sec:mgvu-constructor-root}, we define our new unification algorithm that reduces the total execution time.
Our experiments in Section~\ref{sec:exp} show that this improved unification algorithm works well in practice. 
We conclude in Section~\ref{sec:conc}.

\section{Preliminaries}\label{sec:preliminaries}

We follow the classical notation and terminology from \cite{Terese03} for term rewriting, from \cite{BS-handbook00} for unification, and from \cite{Meseguer92} for rewriting logic and order-sorted notions.

We assume an order-sorted signature $\mathsf{\Sigma} = (S, \leq, \Sigma)$ with a poset of sorts $(S, \leq)$. The poset $(\sort{S},\leq)$ of sorts for $\Symbols$ is partitioned into equivalence classes, called \emph{connected components}, by the equivalence relation $(\leq \cup \geq)^+$. We assume that each connected component $[\sort{s}]$ has a \emph{top element} under $\leq$, denoted $\top_{[\sort{s}]}$ and called the \emph{top sort} of $[\sort{s}]$. This involves no real loss of generality, since if $[\sort{s}]$ lacks a top sort, it can be easily added. We also assume an $\sort{S}$-sorted family $\Variables=\{\Variables_\sort{s}\}_{\sort{s} \in \sort{S}}$ of disjoint variable sets with each $\Variables_\sort{s}$ countably infinite. $\TermsS{s}$ is the set of terms of sort \sort{s}, and $\GTermsS{s}$ is the set of ground terms of sort \sort{s}. We write $\Terms$ and $\GTerms$ for the corresponding order-sorted term algebras. Given a term $t$, $\var{t}$ denotes the set of variables in $t$.

Positions
are represented by sequences of natural numbers denoting an access path in
the term when viewed as a tree.  The top or root position is denoted by the empty sequence $\rootpos$.
We define the relation $p \leq q$ between positions as
$p \leq p $ for any $p$;
and
$p \leq p.q$ for any $p$ and $q$.
Given $U
\subseteq \Symbols\cup\Variables$, $\occSub{t}{U}$ denotes the set of positions of
a term $t$ that are rooted by symbols or variables in $U$.
The set of positions of a term $t$ is written $\occ{t}$,
and
the set of non-variable positions $\funocc{t}$.
The subterm of $t$
at position $p$
is $\subterm{t}{p}$ and $\replace{t}{p}{u}$ is
the term $t$ where $\subterm{t}{p}$ is
replaced by $u$.

A \textit{substitution} $\sigma\in\Substs$ is a sorted mapping from a finite subset of $\Variables$ to $\Terms$. Substitutions are written as $\sigma=\{X_1 \mapsto t_1,\ldots,X_n \mapsto t_n\}$ where the domain of $\sigma$ is $\domain{\sigma}=\{X_1,\ldots,X_n\}$ and the set of variables introduced by terms $t_1,\ldots,t_n$ is written $\range{\sigma}$. The identity substitution is $\idsubst$. Substitutions are homomorphically extended to $\Terms$. The application of a substitution $\sigma$ to a term $t$ is denoted by $t\sigma$ or $\sigma(t)$. For simplicity, we assume that every substitution is idempotent, i.e., $\sigma$ satisfies $\domain{\sigma}\cap\range{\sigma}=\emptyset$. The restriction of $\sigma$ to a set of variables $V$ is $\subterm{\sigma}{V}$, i.e., $\forall x\in V$, $\subterm{\sigma}{V}(x)=\sigma(x)$ and $\forall x\not\in V$, $\subterm{\sigma}{V}(x)=x$. Composition of two substitutions $\sigma$ and $\sigma'$ is denoted by $\sigma\compose\sigma'$. Combination of two substitutions $\sigma$ and $\sigma'$ such that $\domain{\sigma}\cap\domain{\sigma'}=\emptyset$ is denoted by $\sigma \cup \sigma'$. We call a substitution $\sigma$ a variable \emph{renaming} if there is another substitution $\sigma^{-1}$ such that $(\sigma\sigma^{-1})|_{Dom(\sigma)} = \idsubst$.

A \textit{$\Symbols$-equation} is an unoriented pair $t = t'$, where $t,t' \in \TermsS{s}$ for some sort $\sort{s}\in\sort{S}$. An \emph{equational theory} $(\Symbols,E)$ is a pair with $\Symbols$ an order-sorted signature and $E$ a set of $\Symbols$-equations. Given $\Symbols$ and a set $E$ of $\Symbols$-equations, order-sorted equational logic induces a congruence relation $\congr{E}$ on terms $t,t' \in \Terms$ (see~\cite{Meseguer97}). 
We say $\sigma_1 \congr{E} \sigma_2$ iff $\sigma_1(x) \congr{E} \sigma_2(x)$ for any variable $x$.
Throughout this paper we assume that $\GTermsS{s}\neq\emptyset$ for every sort \sort{s}, because this affords a simpler deduction system. An equational theory $(\Symbols,E)$ is \emph{regular} if for each $t = t'$ in $E$, we have $\var{t} = \var{t'}$. An equational theory $(\Symbols,E)$ is \emph{linear} if for each $t = t'$ in $E$, each variable occurs only once in $t$ and in $t'$. An equational theory $(\Symbols,E)$ is \textit{sort-preserving} if for each $t = t'$ in $E$, each sort \sort{s}, and each substitution $\sigma$, we have $t \sigma \in \TermsS{s}$ iff $t' \sigma \in \TermsS{s}$. An equational theory $(\Symbols,E)$ is \emph{defined using top sorts} if for each equation $t = t'$ in $E$, all variables in $\var{t}$ and $\var{t'}$ have a top sort. Given two terms $t$ and $t'$, we say $t$ is more general than $t'$, denoted as $t \sqsupseteq_{E} t'$, if there is a substitution $\eta$ such that $t\eta \congr{E} t'$. Similarly, given two substitutions $\sigma$ and $\rho$, we say $\sigma$ is more general than $\rho$ for a set $W$ of variables, denoted as $\subterm{\sigma}{W} \sqsupseteq_{E} \subterm{\rho}{W}$, if there is a substitution $\eta$ such that $\subterm{(\sigma\compose\eta)}{W} \congr{E} \subterm{\rho}{W}$. The $\sqsupseteq_{E}$ relation induces an equivalence relation $\simeq_{E}$, i.e., $t \simeq_{E} t'$ iff $t \sqsupseteq_{E} t'$ and $t \sqsubseteq_{E} t'$.

An \textit{$E$-unifier} for a $\Symbols$-equation $t = t'$ is a substitution $\sigma$ such that $t\sigma \congr{E} t'\sigma$. For $\var{t}\cup\var{t'} \subseteq W$, a set of substitutions $\csuV{t = t'}{W}{E}$ is said to be a \textit{complete} set of unifiers for the equality $t = t'$ modulo $E$ away from $W$ iff: (i) each $\sigma \in \csuV{t = t'}{W}{E}$ is an $E$-unifier of $t = t'$; (ii) for any $E$-unifier $\rho$ of $t = t'$ there is a $\sigma \in \csuV{t=t'}{W}{E}$ such that $\subterm{\sigma}{W} \sqsupseteq_{E} \subterm{\rho}{W}$; and (iii) for all $\sigma \in \csuV{t=t'}{W}{E}$, $\domain{\sigma} \subseteq (\var{t}\cup\var{t'})$ and $\range{\sigma} \cap W = \emptyset$. Given a conjunction $\Gamma$ of equations, a set $U$ of $E$-unifiers of $\Gamma$ is said to be \textit{minimal} if it is complete and for all distinct elements $\sigma$ and $\sigma'$ in $U$, $\sigma \sqsupseteq_E \sigma'$ implies $\sigma \congr{E} \sigma'$. A unification algorithm is said to be \textit{finitary} and complete if it always terminates after generating a finite and complete set of unifiers. A unification algorithm is said to be \textit{minimal} and complete if it always returns a minimal and complete set of unifiers.

A \textit{rewrite rule} is an oriented pair $l \to r$, where $l \not\in \Variables$ and $l,r \in \TermsS{s}$ for some sort $\sort{s}\in\sort{S}$. An \textit{(unconditional) order-sorted rewrite theory} is a triple $(\Symbols,E,R)$ with $\Symbols$ an order-sorted signature, $E$ a set of $\Symbols$-equations, and $R$ a set of rewrite rules. The set $R$ of rules is \textit{sort-decreasing} if for each $t \rightarrow t'$ in $R$, each $\sort{s} \in \sort{S}$, and each substitution $\sigma$, $t'\sigma \in \TermsS{s}$ implies $t\sigma \in \TermsS{s}$. The rewriting relation on $\Terms$, written $t \rewrite{R} t'$ 
holds between $t$ and $t'$ iff there exist $p \in \funocc{t}$, $l \to r\in R$ and a substitution $\sigma$, such that $\subterm{t}{p} = l\sigma$, and $t' = \replace{t}{p}{r\sigma}$. The relation $\rewrite{R/E}$ on $\Terms$ is ${\congr{E} \composeRel\rewrite{R}\composeRel\congr{E}}$. The transitive (resp. transitive and reflexive) closure of $\rewrite{R/E}$ is denoted $\rewrite{R/E}^+$ (resp. $\rewrites{R/E}$). 

Reducibility of $\rewrite{R/E}$ is undecidable in general since $E$-congruence classes can be arbitrarily large. Therefore, $R/E$-rewriting is usually implemented by $R,E$-rewriting under some conditions on $R$ and $E$ such as confluence, termination, and coherence (see~\cite{JK86,Meseguer17,Meseguer20}). A relation $\rewrite{R,E}$ on $\Terms$ is defined as: $t \rewrite{R,E} t'$ 
iff there is a non-variable position $p \in \funocc{t}$, a rule $l \to r$ in $R$, and a substitution $\sigma$ such that $\subterm{t}{p} \congr{E} l\sigma$ and $t' = \replace{t}{p}{r\sigma}$. The narrowing relation $\narrow{}{R,E}$ on $\Terms$ is defined as: $t \narrow{\sigma}{R,E} t'$ 
iff there is a non-variable position $p \in \funocc{t}$, a rule $l \to r$ in $R$, and a substitution $\sigma$ such that $\subterm{t}{p}\sigma \congr{E} l\sigma$ and $t' = (\replace{t}{p}{r})\sigma$.
We call $(\Symbols,B,E)$ a \emph{decomposition} of an order-sorted equational theory ${(\Symbols,E\uplus B)}$ if $B$ is regular, linear, sort-preserving, defined using top sorts, and has a finitary and complete unification algorithm, 
and equations $E$ are oriented into rules $\overrightarrow{E}$ such that they are sort-decreasing and \emph{convergent}, i.e., confluent, terminating, and strictly coherent modulo $B$  \cite{DM12,LM16,Meseguer17}. The irreducible version of a term $t$ is denoted by $t\norm{E,B}$.

Given a decomposition $(\Symbols,B,E)$ of an equational theory and a term $t$, a pair $(t',\theta)$ of a term $t'$ and a substitution $\theta$ is an $E,B$-\emph{variant} (or just a variant) of $t$ if $t\theta\norm{E,B} \congr{B} t'$ and $\theta\norm{E,B} \congr{B} \theta$ ~\cite{CD05,ESM12}. A \emph{complete set of $E,B$-variants}~\cite{ESM12} (up to renaming) of a term $t$ is a subset, denoted by $\sem{t}$, of the set of all $E,B$-variants of $t$ such that, for each $E,B$-variant $(t',\sigma)$ of $t$, there is an $E,B$-variant $(t'', \theta) \in \sem{t}$ such that $(t'',\theta) \sqsupseteq_{E,B} (t',\sigma)$, i.e., there is a substitution $\rho$ such that $t' \congr{E} t''\rho$ and $\restrict{\sigma}{\var{t}} =_{E} \restrict{(\theta\rho)}{\var{t}}$. A decomposition $(\Symbols,B,E)$ has the \emph{finite variant property} (FVP)~\cite{ESM12} (also called a \emph{finite variant decomposition}) iff for each $\Symbols$-term $t$, there exists a complete and finite set $\sem{t}$ of variants of $t$. Note that whether a decomposition has the finite variant property is undecidable~\cite{BGLN13}, but a technique based on the dependency pair framework has been developed in \cite{ESM12} and a semi-decision procedure that works well in practice is available in~\cite{CME14tr}.

\section{Variant-based Equational Unification in Maude 3.0}\label{sec:eq_unification}

Rewriting logic \cite{Meseguer92} is  a flexible semantic framework within which  different concurrent systems can be naturally specified 
(see \cite{Meseguer12}). Rewriting Logic is efficiently implemented in the high-performance system Maude~\cite{maude-manual}, which has itself a formal environment of verification tools thanks to its reflective capabilities (see \cite{Maude07,Meseguer12}).

Maude 3.0 offers quite sophisticated symbolic capabilities (see \cite{Meseguer18-wollic} and references therein). 
Among these symbolic features, equational unification \cite{maude-manual} is a twofold achievement.
On the one hand, Maude provides an order-sorted equational unification command for any combination of symbols having any combination of 
associativity, commutativity, and identity \cite{DEEM+20}. This is remarkable, since there is no other system with such an advanced unification algorithm.
On the other hand, a narrowing-based equational unification algorithm relying on the concept of the \emph{variants} \cite{CD05} of a term is also available. A variant of a term $t$ is a pair consisting of a substitution $\sigma$ and the canonical form of $t\sigma$. Narrowing was proved to be complete for unification in \cite{JKK83}, but variant-based unification is decidable when the equational theory satisfies the \emph{finite variant property} \cite{CD05,ESM12}.
The finite variant property has become an essential  property in some research areas, such as cryptographic protocol analysis,
where 
Maude-NPA~\cite{EMM09},
Tamarin~\cite{DHRS18} 
and AKISS~\cite{BDGK17}
rely on the different unification features of Maude.

Let us make explicit the relation between variants and equational unification. First, we define the intersection of two sets of variants. Without loss of generality, we assume in this paper that each variant pair $(t',\sigma)$ of a term $t$ uses new freshly generated variables.

\begin{definition}[Variant Intersection]{\rm\cite{ESM12}}\label{def:cap}
Given a decomposition $(\Symbols,B,E)$ of an equational theory, two $\Symbols$-terms $t_1$ and $t_2$ such that $W_\cap = \var{t_1}\cap\var{t_2}$ and $W_\cup = \var{t_1}\cup\var{t_2}$, and two sets $V_1$ and $V_2$ of variants of $t_1$ and $t_2$, respectively, we define 
$V_1 \cap V_2 = \{(u_1\sigma,\theta_1\sigma \cup \theta_2\sigma \cup \sigma) \mid (u_1,\theta_1) \in V_1 \wedge (u_2,\theta_2) \in V_2 \wedge 
\exists \sigma: \sigma \in \csuV{u_1 = u_2}{W_\cup}{B} 
\wedge 
\restrict{(\theta_1\sigma)}{W_\cap} 
\congr{B} 
\restrict{(\theta_2\sigma)}{W_\cap}
\}
$.
\end{definition}

Then, we define variant-based unification as the computation of the variants of the two terms in a unification problem and their intersection.

\begin{corollary}[Finitary $\caE$-unification]{\rm\cite{ESM12}}
\label{cor:finitary-unification}
Let $(\Symbols,B,E)$ be a finite variant decomposition of an equational theory. Given two terms $t,t'$, the set 
$\csuV{t = t'}{\cap}{E\cup B}=\{\theta \mid (w,\theta)\in\sem{t}\cap\sem{t'}\}$ is a \emph{finite and complete} set of unifiers for $t = t'$.
\end{corollary}

The most recent version 3.0 of Maude \cite{maude-manual} incorporates variant-based unification based on the folding variant narrowing strategy \cite{ESM12}. First, there exists a variant generation command of the form: 

\noindent

{\small
\begin{verbatim}
  get variants [ n ] in ModId : Term .
\end{verbatim}
}

\noindent
where $n$ is an optional argument providing a bound on the number of variants requested, so that if the cardinality of the set of variants is greater than the specified bound, the variants beyond that bound are omitted; and \texttt{ModId} is the identifier of the module where the command takes place. Second, there exists a variant-based unification command of the form: 

\noindent

{\small
\begin{verbatim}
  variant unify [ n ] in ModId : T1 =? T1' /\ ... /\  Tk =? Tk' .
\end{verbatim}
}

\noindent
where $k\geq 1$ and $n$ is an optional argument providing a bound on the number of unifiers requested, so that if there are more unifiers, those beyond that bound are omitted; and \texttt{ModId} is the identifier of the module where the command takes place.

\begin{example}\label{ex:xor}
Consider the following equational theory for exclusive-or that assumes three extra constants \verb+a+, \verb+b+, and \verb+c+. 
The second equation is necessary for coherence modulo $AC$.

{\small
\begin{verbatim}
  fmod EXCLUSIVE-OR is 
    sorts Elem EXor . subsort Elem < EXor .
    ops a b c : -> Elem . op mt : -> EXor . op _*_ : EXor EXor -> EXor [assoc comm] .
    vars X Y Z U V : [EXor] .
    eq [idem] :     X * X = mt    [variant] .
    eq [idem-Coh] : X * X * Z = Z [variant] .
    eq [id] :       X * mt = X    [variant] .
  endfm
\end{verbatim}
}

\noindent
The attribute \verb+variant+ specifies that these equations will be used for variant-based unification. Since this theory has the finite variant property (see \cite{CD05,ESM12}), given the term \verb!X * Y! it is easy to verify that there are seven most general variants.

{\small
\begin{verbatim}
  Maude> get variants in EXCLUSIVE-OR : X * Y .

  Variant #1                             ...           Variant #7
  [EXor]: #1:[EXor] * #2:[EXor]          ...           [EXor]: %1:[EXor]
  X --> #1:[EXor]                        ...           X --> %1:[EXor]
  Y --> #2:[EXor]                        ...           Y --> mt
\end{verbatim}
}

\noindent  Note that Maude produces fresh variables of the form {\small\textit{\texttt{\#n:Sort}}} or {\small\textit{\texttt{\%n:Sort}}} using two different counters (see \cite{maude-manual} for details).
When we consider a variant unification problem between terms $X * Y$ and $U * V$, there are $57$ unifiers:

{\small
\begin{verbatim}
  Maude> variant unify in EXCLUSIVE-OR : X * Y =? U * V  .
  Unifier #1                            ...           Unifier #2
  X --> %1:[EXor] * %3:[EXor]           ...           X --> %1:[EXor] * %3:[EXor]
  Y --> %2:[EXor] * %4:[EXor]           ...           Y --> %2:[EXor]
  V --> %1:[EXor] * %2:[EXor]           ...           V --> %1:[EXor] * %2:[EXor]
  U --> %3:[EXor] * %4:[EXor]           ...           U --> %3:[EXor]
\end{verbatim}
}

\end{example}

However, this variant-based unification algorithm may compute many more unifiers than the necessary or may not be able to stop immediately. 
For instance, it is well-known that unification in the exclusive-or theory is unitary, i.e., there exists only one most general unifier modulo exclusive-or
\cite{KN87}.
For the unification problem $X * Y \? U * V$ of Example~\ref{ex:xor}, the  most general unifier w.r.t. $\sqsupseteq_{E\cup B}$ is 
$\{X \mapsto Y * U * V\}$,
which should be appropriately written as 
$\sigma=\{X \mapsto Y' * U' * V', Y \mapsto Y', U \mapsto U', V \mapsto V'\}.$
Note that 
$\{Y \mapsto X * U * V\}$,
$\{U \mapsto Y * X * V\}$,
and
$\{V \mapsto Y * U * X\}$ are equivalent to the former unifier w.r.t. $\sqsupseteq_{E\cup B}$
by composing $\sigma$ with, respectively,
$\rho_1=\{Y' \mapsto X'' * U'' * V'',X' \mapsto X'', U' \mapsto U'', V' \mapsto V''\}$,
$\rho_2=\{U' \mapsto Y'' * X'' * V'',X' \mapsto X'', Y' \mapsto Y'', V' \mapsto V''\}$,
and
$\rho_3=\{V' \mapsto Y'' * U'' * X'',X' \mapsto X'', U' \mapsto U'', Y' \mapsto Y''\}$.
Similarly,
$\{X \mapsto U, Y \mapsto V\}$ 
and
$\{X \mapsto V, Y \mapsto U\}$ 
are equivalent to all the previous ones.

Furthermore, since the variants of both terms are generated by Corollary~\ref{cor:finitary-unification}, 
there may be very simple unification problems such as $X \? t$ where the generation of the variants of $t$ is unnecessary.  
For example, when unifying terms $X$ and $U * V$, the variants of $U * V$ are generated

{\small
\begin{verbatim}
  Maude> variant unify in EXCLUSIVE-OR : X =? U * V  .

  Unifier #1                      Unifier #2           Unifier #3
  X --> %1:[EXor] * %2:[EXor]     X --> mt             X --> #2:[EXor] * #3:[EXor]
  V --> %1:[EXor]                 V --> #1:[EXor]      V --> #1:[EXor] * #2:[EXor]
  U --> %2:[EXor]                 U --> #1:[EXor]      U --> #1:[EXor] * #3:[EXor]
		
  Unifier #4	                   Unifier #5	                   Unifier #6
  X --> #1:[EXor]	               X --> #1:[EXor]	               X --> #1:[EXor]
  V --> #1:[EXor] * #2:[EXor]	   V --> #2:[EXor]	               V --> mt
  U --> #2:[EXor]	               U --> #1:[EXor] * #2:[EXor]	   U --> #1:[EXor]
		
  Unifier #7		
  X --> #1:[EXor]		
  V --> #1:[EXor]		
  U --> mt
\end{verbatim}
}

\noindent but it is clear that the simplest, most general unifier is $\{X \mapsto U * V\}$.
In \cite{ES19}, a new procedure to reduce the number of variant unifiers in situations like this was developed.
We showed that this new procedure pays off in practice using both the exclusive-or and the abelian group equational theories.

\section{Constructor-Root Variant-based Unification}\label{sec:mgvu-constructor-root}

Both the ``logic" notion of a functor
and 
the ``functional" 
notion of a constructor refer to a symbol not appearing in the root position of the left-hand side of any predicate or equation.
This notion of constructor
allows to split a signature $\Symbols$
as a disjoint union $\Symbols = \DSymbols \uplus \CSymbols$
where $\DSymbols$ are called \emph{defined} symbols and $\CSymbols$ are called \emph{constructor} symbols.
In a decomposition $(\Symbols,B,E)$,
the \emph{canonical term algebra} 
$\textit{Can}_{\Symbols/(E,B)}=\{t\norm{E,B} \mid t \in\GTerms\}$
is typically made of constructor terms,
but this more general notion of constructor differs from the ``logic" and ``functional" notions.
A decomposition $(\Symbols,B,E)$ \emph{protects} a \emph{constructor decomposition} $(\CSymbols, B_\CSymbols, E_\CSymbols)$ iff $\CSymbols \subseteq \Symbols$, $B_\CSymbols \subseteq B$, and $E_\CSymbols \subseteq E$,
and for all $t,t' \in \TermsOn{\CSymbols}{\Variables}{}{}{}$ we have: 
(i) $t \congr{B_\CSymbols} t' \iff t \congr{B} t'$, 
(ii) $t = t\norm{E_\CSymbols,B_\CSymbols} \iff t = t\norm{E,B}$, and 
(iii) $\textit{Can}_{\CSymbols/(E_\CSymbols,B_\CSymbols)} = \textit{Can}_{\Symbols/(E,B)}|_{\CSymbols}$.
A \emph{constructor decomposition} $(\CSymbols, B_\CSymbols,\emptyset)$ is called \emph{free}.
For instance, 
the modular exponentiation property typical of Diffie-Hellman protocols
is defined using 
two versions of the exponentiation operator
and an auxiliary associative-commutative symbol $*$ for exponents so that $(z^x)^y = (z^y)^x=z^{x*y}$. 
Note that, in the lefthand side of the equation, the outermost exponentiation operator is defined, whereas
the innermost  exponentiation operator is constructor.

{\small
\begin{verbatim}
fmod DH-CFVP is
  sorts Exp Elem ElemSet Gen .  subsort Elem < ElemSet .
  ops a b c : -> Elem [ctor] .
  op exp : Gen ElemSet -> Exp [ctor] .
  op exp : Exp ElemSet -> Exp .
  op _*_ : ElemSet ElemSet -> ElemSet [assoc comm ctor] .
  var X : Gen . vars Y Z : ElemSet .
  eq exp(exp(X,Y),Z) = exp(X,Y * Z) [variant] .
endfm
\end{verbatim}
}

\noindent
Note that it may not always be possible to provide a (free) constructor decomposition, such as Example~\ref{ex:xor}
where the exclusive-or symbol works both as defined and constructor
(see \cite{maude-manual} for a detailed discussion).
However, it is common to combine an equational theory with many different additional constructor symbols,
as shown in Section~\ref{sec:exp}.

The notion of a constructor symbol has not yet been integrated into
the variant-based equational unification procedure of Maude.
An integration of the notion of constructor involves two challenges.
On the one hand, when we consider the variant unification problem above between terms $X$ and $U * V$,
the fast unification algorithm of \cite{ES19} is able to return only one unifier but still generates all the variants of term $U * V$,
unnecessarily consuming  resources. 
On the other hand, a unification problem between terms $f(X * Y)$ and $g(U * V)$ where $f$ and $g$ are different constructor symbols 
forces the generation of
all the variants of the terms $X* Y$ and $U * V$ wasting resources.
Let us consider a unification problem $C_1[X] \? C_2[t]$
where
both $C_1$ and $C_2$ are made of constructor symbols and
either there exists $\sigma$ s.t. $C_1[\Box]\sigma =_{B} C_2[\Box]\sigma$ 
or there is no such $\sigma$. 

\begin{definition}[Constructor-root Position]
Given a decomposition $(\Symbols,B,E)$ 
protecting a free constructor decomposition $(\CSymbols, B_\CSymbols,\emptyset)$
and given
a $\Symbols$-term $t$ and a position $p\in\Pos(t)$,
we say $p$ is a \emph{constructor-root} position in $t$
if 
for all $q < p$, $root(t|_{q})\in\CSymbols$.
\end{definition}

\begin{definition}[Constructor-root Variable]
Given a $\Symbols$-term $t$ and a variable $x$,
we say $x$ is a \emph{constructor-root} variable in $t$
if 
for all 
$p\in\occSub{t}{x}$, $p$ is constructor-root in $t$.
\end{definition}

First, we define the case 
when there exists $\sigma$ s.t. $C_1[\Box]\sigma =_{B} C_2[\Box]\sigma$.
Intuitively, a variant unifier $\sigma$ of $t_1$ and $t_2$ is constructor-root if each variable in $\range{\sigma}$ 
is under a constructor-root variable of $t_1$ and $t_2$.

\begin{definition}[Constructor-root Variant Unifier]\label{def-new}
Given a decomposition $(\Symbols,B,E)$ 
protecting a free constructor decomposition $(\CSymbols, B_\CSymbols, \emptyset)$,
two $\Symbols$-terms $t_1$ and $t_2$ s.t.
$W_\cap = \var{t_1}\cap\var{t_2}$,
$W_\cup = \var{t_1}\cup\var{t_2}$,
$(u_1,\theta_1) \in \sem{t_1}$,
$(u_2,\theta_2) \in \sem{t_2}$,
and
$\sigma \in \csuV{u_1 = u_2}{W_\cup}{B}$ 
s.t.
$\restrict{(\theta_1\sigma)}{W_\cap} 
\congr{B} 
\restrict{(\theta_2\sigma)}{W_\cap}$,
the unifier 
$(\theta_1\cup\theta_2)\sigma$
is called \emph{constructor-root}
if
for each $x \mapsto t \in \sigma$,
either 
(i)
$x \mapsto t$ is a variable renaming,
(ii)
$x$ is a constructor-root variable in $u_1$ and $u_2$,
or 
(iii)
for each $x' \mapsto t'\in \sigma\setminus\{x \mapsto t\}$ 
(and there exists at least one such binding) s.t. 
$t' =_{B} C[t]$,
then
$x'$ is a constructor-root variable in $u_1$ and $u_2$.
\end{definition}

Let us motivate the usefulness of a constructor-root unifier.
Given the unification problem $\mathtt{X}\ \?\ \mathtt{V * U}$ above, 
the  unifier
$\{\mathtt{X} \mapsto \mathtt{\%1 * \%2}$, $\mathtt{V}\mapsto\mathtt{\%1}$, $\mathtt{U}\mapsto\mathtt{\%2}\}$
is constructor-root, since \texttt{X} is a constructor-root variable in the left unificand
and \texttt{V} and \texttt{U} are not constructor-root variables but the variables $\%1$ and $\%2$ used in the bindings of \texttt{V} and \texttt{U}
appear  in the binding of \texttt{X}.
Hence, we can safely avoid the generation of the variants of $V * U$.
Note that
the  unifier
$\{\mathtt{X} \mapsto \mathtt{mt}$, $\mathtt{V}\mapsto\mathtt{\%1}$, $\mathtt{U}\mapsto\mathtt{\%1}\}$
is not constructor-root 
because 
$\mathtt{V}$ and $\mathtt{U}$ are not constructor-root variables
and
for the bindings $\mathtt{V}\mapsto\mathtt{\%1}$ and $\mathtt{U}\mapsto\mathtt{\%1}$
there is no other binding $x' \mapsto t'$ such that $\mathtt{\%1}$ is a subterm of $t'$
and $x'$ is a constructor-root variable.

\begin{lemma}[Constructor-root Variant Unifier]\label{lem-new}
Given a decomposition $(\Symbols,B,E)$ 
protecting a free constructor decomposition $(\CSymbols, B_\CSymbols, \emptyset)$,
two $\Symbols$-terms $t_1$ and $t_2$ s.t.
$W_\cap = \var{t_1}\cap\var{t_2}$,
$W_\cup = \var{t_1}\cup\var{t_2}$,
$(u_1,\theta_1) \in \sem{t_1}$,
$(u_2,\theta_2) \in \sem{t_2}$,
and
a {\em constructor-root variant unifier}
$\sigma \in \csuV{u_1 = u_2}{W_\cup}{B}$
s.t. 
$\restrict{(\theta_1\sigma)}{W_\cap} 
\congr{B} 
\restrict{(\theta_2\sigma)}{W_\cap}$,
then 
$\forall (u'_1,\theta'_1) \in \sem{t_1}$ s.t. $(u'_1,\rho) \in \sem{u_1}$ and $\restrict{\theta'_1}{W_\cup} \congr{B} \restrict{\theta_1\rho}{W_\cup}$,
if there exists 
$\sigma' \in \csuV{u'_1 = u_2}{W_\cup}{B}$ s.t.
$\restrict{(\theta'_1\sigma')}{W_\cap} 
\congr{B} 
\restrict{(\theta_2\sigma')}{W_\cap}$,
then
$\restrict{((\theta_1\cup\theta_2)\sigma)}{W_\cup}$ and $\restrict{((\theta'_1\cup\theta_2)\sigma')}{W_\cup}$ are both equational unifiers of $t_1$ and $t_2$
but 
$\restrict{((\theta_1\cup\theta_2)\sigma)}{W_\cup} \sqsupseteq_{E\cup B} \restrict{((\theta'_1\cup\theta_2)\sigma')}{W_\cup}$.
Similarly for any $(u'_2,\theta'_2) \in \sem{t_2}$.
\end{lemma}
\begin{proof}
By contradiction.
Let us assume 
$\exists\sigma' \in \csuV{u'_1 = u_2}{W_\cup}{B}$
s.t. 
$\restrict{((\theta_1\cup\theta_2)\sigma)}{W_\cup} \not\sqsupseteq_{E\cup B} \restrict{((\theta'_1\cup\theta_2)\sigma')}{W_\cup}$.
First, $\restrict{\theta'_1}{t_1} = \restrict{\theta_1}{t_1}\restrict{\rho}{u_1}$ and thus the difference is in $\restrict{\sigma}{u_2}$ and $\restrict{\sigma'}{u_2}$.
By the constructor-root property,
$\forall x \mapsto t \in \restrict{\sigma}{u_2}$
either $x$ is a constructor-root variable in $u_2$
or 
$\forall x' \mapsto t'\in \restrict{(\sigma\setminus\{ x \mapsto t\})}{u_2}$ s.t. 
$t' =_{B} C[t]$,
$x'$ is a constructor-root variable in $u_2$.
But then $\forall y\in\var{u_2}$,
there exist $y \mapsto w_1\in\restrict{\sigma}{u_2}$ and $y \mapsto w_2\in\restrict{\sigma'}{u_2}$
and $(w_2,\rho) \in \sem{w_1}$,
i.e.,
$\restrict{\sigma}{u_2} \sqsupseteq_{E\cup B} \restrict{\sigma'}{u_2}$, which contradicts the assumption.
\end{proof}

We define the case 
when there is no $\sigma$ s.t. $C_1[\Box]\sigma =_{B} C_2[\Box]\sigma$.
Intuitively, two terms that form a constructor-root failure pair will never unify despite any further variant computation.

\begin{definition}[Constructor-root Failure Pair]\label{def-new-neg}
Given a decomposition $(\Symbols,B,E)$ 
protecting a free constructor decomposition $(\CSymbols, B_\CSymbols, \emptyset)$,
two $\Symbols$-terms $t_1$ and $t_2$ s.t.
$W_\cap = \var{t_1}\cap\var{t_2}$,
$W_\cup = \var{t_1}\cup\var{t_2}$,
$(u_1,\theta_1) \in \sem{t_1}$,
and
$(u_2,\theta_2) \in \sem{t_2}$,
the pair $(u_1,u_2)$ is a \emph{constructor-root failure pair}
if 
$\csuV{u_1 = u_2}{W_\cup}{B}=\emptyset$ 
and
there exists two constructor contexts $C_1[\Box,\ldots,\Box]$ and $C_2[\Box,\ldots,\Box]$,
terms $v_1,\ldots,v_n,w_1,\ldots,w_m$,
and fresh distinct variables $x_1,\ldots,x_n,y_1,\ldots,y_m$
s.t.
$u_1 =_B C_1[v_1,\ldots,v_n]$, 
$u_2 =_B C_2[w_1,\ldots,w_m]$, 
and 
$\csuV{C_1[x_1,\ldots,x_n] = C_2[y_1,\ldots,y_m]}{W_\cup}{B}=\emptyset$ .
\end{definition}

Let us motivate the usefulness of a constructor-root failure pair.
Given the unification problem $\mathtt{f(X * Y)}\ \?\ \mathtt{g(V * U)}$ above
where $f$ and $g$ are different constructor symbols without axioms,
the two terms do not unify modulo the axioms of $*$ but 
neither $f(W)$ and $g(W')$ do.
Hence, we can safely avoid the generation of the variants of $V * U$ and $X * Y$.

\begin{lemma}[Constructor-root Failure Pair]\label{lem-new-neg}
Given a decomposition $(\Symbols,B,E)$ 
protecting a free constructor decomposition $(\CSymbols, B_\CSymbols, \emptyset)$,
two $\Symbols$-terms $t_1$ and $t_2$ s.t.
$W_\cup = \var{t_1}\cup\var{t_2}$,
$(u_1,\theta_1) \in \sem{t_1}$,
$(u_2,\theta_2) \in \sem{t_2}$,
and
$(u_1,u_2)$ is a constructor-root failure term, 
then 
$\csuV{u_1 = u_2}{W_\cup}{E\cup B}=\emptyset$.
\end{lemma}
\begin{proof}
Immediate by Definition~\ref{def-new-neg}.
\end{proof}

These positive and negative stopping criteria, however, become useful only if we do not generate all the variants a priori,
as it is done in Corollary \ref{cor:finitary-unification} as well as the fast unification technique of \cite{ES19}.
Some incremental generation of variants is required. 

\begin{example}\label{ex:fg}
Consider the following theory 
where constructors have the \verb!ctor! attribute.

{\small
\begin{verbatim}
  fmod FASTvsCR is sort S . 
    ops a b c : -> S [ctor] . op s : S -> S [ctor] .
    op g : S S -> S .         op f : S S S -> S .      vars X Y Z W : S .
    eq f(a,X,Y) = s(Y) [variant] .
    eq f(b,X,Y) = g(X,Y) [variant] .
    eq g(c,Y) = s(Y) [variant] .
  endfm
\end{verbatim}
}

\noindent
Consider  the unification problem (a) $f(X,Y,Z) = s(W)$ with only two 
unifiers and its variant generation.

\noindent
\begin{tabular}{c@{}c@{}}
\begin{minipage}[t]{.5\linewidth}
{\small
\begin{verbatim}
variant unify in FG : f(X, Y, Z) =? s(W) .

  Unifier #1        Unifier #2
  X --> a           X --> b
  Y --> #2:S        Y --> c
  Z --> #1:S        Z --> %1:S
  W --> #1:S        W --> %1:S
\end{verbatim}
}
\end{minipage}
&
$$
\xymatrix@R=1.2pc@C=1pc{
&\mathtt{f(X,Y,Z)}\ar@{~>}_{\{\mathtt{X}\mapsto\mathtt{a}\}}[dl]\ar@{~>}_{\{\mathtt{X}\mapsto\mathtt{b}\}}[dr] & \? & \mathtt{s(W)}\\
\mathtt{s(Z)} & & \mathtt{g(Y,Z)}\ar@{~>}_{\{\mathtt{Y}\mapsto\mathtt{c}\}}[d]\\
& & \mathtt{s(Z)}
}
$$
\end{tabular}
\vspace{2ex}

Let us assume we have an expression $\clubsuit$ with a considerably large narrowing tree
and two new unification problems (b) $f(X,Y,\clubsuit) = s(W)$ and (c) $f(X,\clubsuit,Z) = s(W)$.
Note that the unifiers of (a) 
are still valid for (b), whereas only the first unifier of 
(a) 
is valid for (c), 
assuming $\clubsuit$ never narrows into $c$.
Both the variant-based unification command of Maude and the fast command of \cite{ES19}
cannot avoid the computation of $\clubsuit$ in both unification problems (b) and (c).
However, the technique described below is able to avoid the full computation of $\clubsuit$ in (b), since the two unifiers are constructor-root,
although it 
cannot avoid the full computation of $\clubsuit$ in (c).
\end{example}

We extend the notions of constructor-root unifier and constructor-root failure pair
to the pairwise combination of all the variants of a unification problem.

\begin{definition}[Constructor-Root Intersection]
Given a decomposition $(\Symbols,B,E)$ 
protecting a free constructor decomposition $(\CSymbols, B_\CSymbols, \emptyset)$,
two $\Symbols$-terms $t_1$ and $t_2$ such that $W_\cap = \var{t_1}\cap\var{t_2}$ and $W_\cup = \var{t_1}\cup\var{t_2}$, and two sets $V_1$ and $V_2$ of variants of $t_1$ and $t_2$, respectively, we 
say that an intersection $V_1 \cap V_2$ is \emph{constructor-root} if 
for each leaf $(u_1,\theta_1) \in V_1$ (resp. $(u_2,\theta_2) \in V_2$),
and
for each leaf $(u_2,\theta_2) \in V_2$ (resp. $(u_1,\theta_1) \in V_1$) such that 
$\sigma \in \csuV{u_1 = u_2}{W_\cup}{B}$ and
$\restrict{(\theta_1\sigma)}{W_\cap} 
\congr{B} 
\restrict{(\theta_2\sigma)}{W_\cap}$,
we have 
$(\theta_1\cup\theta_2)\sigma$
is \emph{constructor-root}.
\end{definition}

\begin{definition}[Failure Intersection]
Given a decomposition $(\Symbols,B,E)$ 
protecting a free constructor decomposition $(\CSymbols, B_\CSymbols, \emptyset)$
two $\Symbols$-terms $t_1$ and $t_2$ such that $W_\cap = \var{t_1}\cap\var{t_2}$ and $W_\cup = \var{t_1}\cup\var{t_2}$, and two sets $V_1$ and $V_2$ of variants of $t_1$ and $t_2$, respectively, we 
say that an intersection $V_1 \cap V_2$ is a \emph{failure intersection} if 
for each leaf $(u_1,\theta_1) \in V_1$ (resp. $(u_2,\theta_2) \in V_2$),
and
for each leaf $(u_2,\theta_2) \in V_2$ (resp. $(u_1,\theta_1) \in V_1$) such that 
$\sigma \in \csuV{u_1 = u_2}{W_\cup}{B}$ and
$\restrict{(\theta_1\sigma)}{W_\cap} 
\congr{B} 
\restrict{(\theta_2\sigma)}{W_\cap}$,
we have 
the pair $(u_1,u_2)$ is a constructor-root failure pair.
\end{definition}

The following example shows that the folding variant narrowing trees of both terms $t_1,t_2$ of a unification problem $t_1 = t_2$
must be unfolded down to a frontier where all leaves of $t_1$ are tested for unification with all the leaves of $t_2$.

\begin{example}\label{ex:fg2}
Let us consider 
Example~\ref{ex:fg} 
and the unification problem $f(X,Y,Z) = f(U,V,W)$.

{\small
\begin{verbatim}
variant unify in FG : f(X, Y, Z) =? f(U, V, W)  .

  Unifier #1           Unifier #2           Unifier #3           Unifier #4
  X --> %1:S           X --> a              X --> a              X --> b
  Y --> %2:S           Y --> %2:S           Y --> #2:S           Y --> c
  Z --> %3:S           Z --> %1:S           Z --> #1:S           Z --> #1:S
  U --> %1:S           U --> a              U --> b              U --> a
  V --> %2:S           V --> %3:S           V --> c              V --> #2:S
  W --> %3:S           W --> %1:S           W --> #1:S           W --> #1:S
\end{verbatim}
}

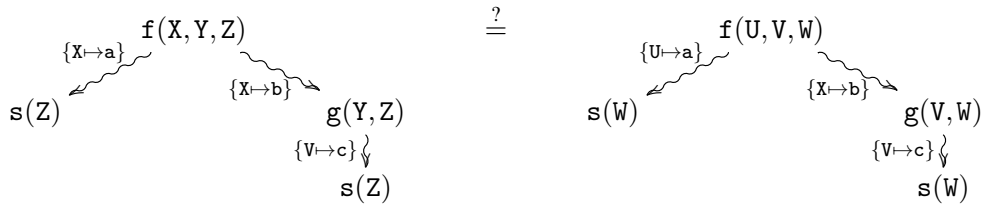
$$
\xymatrix@R=1pc{
&\mathtt{f(X,Y,Z)}\ar@{~>}_{\{\mathtt{X}\mapsto\mathtt{a}\}}[dl]\ar@{~>}_{\{\mathtt{X}\mapsto\mathtt{b}\}}[dr] 
& & \? &
& \mathtt{f(U,V,W)}\ar@{~>}_{\{\mathtt{U}\mapsto\mathtt{a}\}}[dl]\ar@{~>}_{\{\mathtt{X}\mapsto\mathtt{b}\}}[dr] & \\
\mathtt{s(Z)} & & \mathtt{g(Y,Z)}\ar@{~>}_{\{\mathtt{V}\mapsto\mathtt{c}\}}[d] & &
\mathtt{s(W)} & & \mathtt{g(V,W)}\ar@{~>}_{\{\mathtt{V}\mapsto\mathtt{c}\}}[d]
\\
& & \mathtt{s(Z)} & &
& & \mathtt{s(W)}
}
$$

\noindent
The terms at the top position of both narrowing trees clearly unify, but the unifier is not constructor-root, so we must continue expanding both narrowing trees.
The condition that all leaves of $t_1$ are unifiable with all the leaves of $t_2$ is reached only at depth $2$. 
Indeed, if we expand  the left unificand completely but the right unificand only down to the leftmost branch,
then the two leaves of the narrowing tree of the left  unificand unify with the leftmost leaf of the narrowing tree of the right unificand,
but we may miss the last two unifiers reported above if we stop here.
\end{example}
 
We define variant-based unification as the computation of the variants of the two terms in a unification problem. 
We abuse the notation and write 
$\mathbb{P}(\sem{t})$ for the powerset of all the subsets of $\sem{t}$ such that
each $V\in\mathbb{P}(\sem{t})$ corresponds to the variants of a term $t$ associated to a particular narrowing tree produced by the folding variant narrowing strategy from term $t$.
We also write $\csuV{t = t'}{\cap,V_1,V_2}{E\cup B}$ for a version of the unification algorithm of Corollary~\ref{cor:finitary-unification} that uses sets $V_1$ and $V_2$ of variants of $t$ and $t'$, respectively, instead of generating all the variants.

\begin{definition}[Constructor-Root Variant-based Unification]\label{def:mgvu-new}
Let $(\Symbols,B,E)$ be a finite variant decomposition of an equational theory
protecting a free constructor decomposition $(\CSymbols, B_\CSymbols, \emptyset)$.
Given two terms $t,t'$ and two sets of variants $V_1\in\mathbb{P}(\sem{t})$, $V_2\in\mathbb{P}(\sem{t'})$,
the constructor-root variant unifiers are

\noindent
$$\csuV{t = t'}{\overline{\cap}}{E\cup B}
=
\left\{
\begin{array}{@{}c@{\ }l@{}}
\emptyset &\mbox{if } \exists V_1\in\mathbb{P}(\sem{t}), V_2\in\mathbb{P}(\sem{t'}),\\
& \mbox{and they are the smallest sets s.t }\\
& V_1\cap V_2 \mbox{ is a failure intersection}\\
\csuV{t = t'}{\cap,V_1,V_2}{E\cup B} &\mbox{if } \exists V_1\in\mathbb{P}(\sem{t}), V_2\in\mathbb{P}(\sem{t'}),\\
& \mbox{and they are the smallest sets s.t }\\
& V_1\cap V_2 \mbox{ is constructor-root}\\
\csuV{t = t'}{\cap,\sem{t},\sem{t'}}{E\cup B} &\mbox{otherwise}\\
\end{array}\right.
$$

\end{definition}

\begin{proposition}[Constructor-Root Variant-based Unification]\label{prop:mgvu-new}
Let $(\Symbols,B,E)$ be a finite variant decomposition of an equational theory
protecting a free constructor decomposition $(\CSymbols, B_\CSymbols, \emptyset)$.
Given two terms $t,t'$, the set 
$\csuV{t = t'}{\overline{\cap}}{E\cup B}$ is a \emph{finite and complete} set of unifiers for $t = t'$.
\end{proposition}

\begin{proof}
By contradiction.
Let us assume that $\csuV{t = t'}{\overline{\cap}}{E\cup B}$ is not a complete set of unifiers of $t$ and $t'$.
That is, 
there exists a unifier $\rho'\in\csuV{t = t'}{\cap}{E\cup B}$ and
there is no unifier $\rho\in\csuV{t = t'}{\overline{\cap}}{E\cup B}$
s.t. $\rho \sqsupseteq_{E\cup B} \rho'$.
By definition, 
there exist smallest sets $V_1\in\mathbb{P}(\sem{t})$, $V_2\in\mathbb{P}(\sem{t'})$
s.t.
$V_1\cap V_2$ is constructor-root or a failure pair.
The case of a failure pair is immediate by Lemma~\ref{lem-new-neg}.
Since $\rho'\in\csuV{t = t'}{\cap}{E\cup B}$, we have that 
there exists $u_1,u_2,\theta_1,\theta_2,\sigma$ s.t.
$\rho'=\theta_1\sigma \cup \theta_2\sigma \cup \sigma$,
 $(u_1,\theta_1) \in \sem{t}$, $(u_2,\theta_2) \in \sem{t'}$, 
$\sigma \in \csuV{u_1 = u_2}{W_\cup}{B}$,
and
$\restrict{(\theta_1\sigma)}{W_\cap} 
\congr{B} 
\restrict{(\theta_2\sigma)}{W_\cap}$.
Since $V_1\cap V_2$ is constructor-root, there must be two leaves
 $(v_1,\tau_1) \in V_1$, $(v_2,\tau_2) \in V_2$
 and a substitution $\tau_3$ s.t. 
$\tau_3 \in \csuV{v_1 = v_2}{W_\cup}{B}$,
$\restrict{(\tau_1\tau_3)}{W_\cap} 
\congr{B} 
\restrict{(\tau_2\tau_3)}{W_\cap}$,
and
$(\tau_1\cup\tau_2)\tau_3$
is \emph{constructor-root}.
Furthermore, the variant  $(u_1,\theta_1)$ (resp. $(u_2,\theta_2)$) is obtained by further narrowing of $v_1$ (resp. $v_2$),
i.e.,
$(u_1,\tau'_1)\in \sem{v_1}$ and $(u_2,\tau'_2)\in \sem{v_2}$.
But then the conclusion follows, since the statement is 
$\restrict{((\tau_1\cup\tau_2)\tau_3)}{W_\cup} \sqsupseteq_{E\cup B} \restrict{((\theta_1\cup\theta_2)\sigma)}{W_\cup}$.
\end{proof}

\section{Experimental Evaluation}\label{sec:exp}

We have performed some experiments with the constructor-root variant-based unification, which are
available at \url{http://safe-tools.dsic.upv.es/cr-mgvu}.

All the experiments were conducted on a PC with a 3.3GHz Intel Xeon E5-1660 and 64GB RAM. We created a battery of 15 different unification problems for both the {\it exclusive-or} and the {\it abelian group} theories. 
These are among the most complicated cryptographic theories in protocol analysis
that Maude-NPA~\cite{EMM09},
Tamarin~\cite{DHRS18} 
and AKISS~\cite{BDGK17} can hardly handle.
Indeed, 
the {\it exclusive-or} and the {\it abelian group} theories cannot be specified in Maude using constructor symbols
and we introduce arbitrary constructors $f_1,f_2,f_3,f_4,f_5$, where the subindex indicates the number of arguments. 
This is a common situation in crypto protocol analysis where the cryptographic properties are combined
with many different additional constructor symbols. 
Experiments using other cryptographic theories, such as Diffie-Hellmann exponentiation, or more traditional programs, 
such as manipulating complex data structures, could also have been included but were discarded 
because the improvement is less remarkable.

For each problem and theory, we computed: 
(i) the unifiers using the standard {\tt variant\;unify} command provided by the C++ core system of Maude; 
(ii) the unifiers using the algorithm $\csuV{t = t'}{\doublecap}{E\cup B}$ of \cite{ES19}; 
(iii) the unifiers using the algorithm $\csuV{t = t'}{\overline{\cap}}{E\cup B}$ of Definition~\ref{def:mgvu-new},
and
(iv) the unifiers using the algorithm 
$\csuV{t = t'}{\overline{\doublecap}}{E\cup B}$ 
obtained from
$\csuV{t = t'}{\overline{\cap}}{E\cup B}$ 
 by 
replacing $\csuV{t = t'}{\cap,V_1,V_2}{E\cup B}$ with $\csuV{t = t'}{\doublecap,V_1,V_2}{E\cup B}$ from \cite{ES19}.
Note that (ii), (iii), and (iv) are implemented at the metalevel of Maude. 
We measured both the number of computed unifiers and the time required for their computation.

Table~\ref{tab:xor} (resp. Table~\ref{tab:ag}) shows the results obtained for the {\it exclusive-or} (resp. {\it abelian group)} theory. \emph{T/O} indicates that a generous $4$ hours \emph{timeout} was reached without any response. The first column describes the unification problem, 
while the following $\#_{\mathit{maude}}$, $\#_{\mathit{fast}}$, $\#_{\mathit{cr}}$ and $\#_{\mathit{cr+fast}}$ columns show the number of computed unifiers for 
all four unification algorithms (i), (ii), (iii), (iv) described above
and the columns
$\mathcal{T}_{\mathit{maude}}$, $\mathcal{T}_{\mathit{fast}}$, $\mathcal{T}_{\mathit{cr}}$ and $\mathcal{T}_{\mathit{cr+fast}}$ 
show the time (in milliseconds) required to execute the unification command.
Note that it is unfair to compare the performance between compiled code ($\mathcal{T}_{\mathit{maude}}$ column) and interpreted code
($\mathcal{T}_{\mathit{fast}}$, $\mathcal{T}_{\mathit{cr}}$ and $\mathcal{T}_{\mathit{cr+fast}}$ columns), i.e., the C++ core system of Maude and a Maude program using Maude's metalevel. However, our constructor-root unification algorithm is able to beat the compiled code in almost all the unification problems.

\rowcolors{1}{gray!10}{white}
\begin{table}[t]
	\centering
	\scriptsize
	{\setlength{\tabcolsep}{0.5em}
	\begin{tabular*}{\textwidth}{|c|>{\centering}p{7.5cm}|r|r|r|r|r|r|r|r|}
		\hhline{-|-|-|-|-|-|-|-|-|-|}
		\rowcolor{gray!50}
		\multicolumn{2}{|c|}{\it Unification problem}
		&\multicolumn{1}{c|}{$\#_{\mathit{maude}}$}
		&\multicolumn{1}{c|}{$\#_{\mathit{fast}}$}
		&\multicolumn{1}{c|}{$\#_{\mathit{cr}}$}
		&\multicolumn{1}{c|}{$\#_{\mathit{cr+fast}}$}
		&\multicolumn{1}{c|}{$\mathcal{T}_{\mathit{maude}}$} 
		&\multicolumn{1}{c|}{$\mathcal{T}_{\mathit{fast}}$} 
		&\multicolumn{1}{c|}{$\mathcal{T}_{\mathit{cr}}$} 
		&\multicolumn{1}{c|}{$\mathcal{T}_{\mathit{cr+fast}}$}\\
		\hhline{-|-|-|-|-|-|-|-|-|-|}
        $P_{1}$ & $ V_1 \? V_2 * V_3 * V_4$                                         &57     &1      &1   &1      &50      &95     &1    &1          \\
        $P_{2}$ & $ V_1 \? f_3(V_2 * V_3,f_1(V_3 * V_4),f_2(V_2,f_1(V_4)))$         &61     &1      &1   &1      &172     &243     &2    &2          \\
        $P_{3}$ & $ V_1 * V_2 \? V_3 * V_4$                                         &57     &8      &41   &8     &9       &89       &83   &131        \\
        $P_{4}$ & $ V_1 * V_2 \? f_2(V_3,f_1(V_4 * V_5))$                     &28    &4      &4   &4      &12     &18     &8 &10       \\
        $P_{5}$ & $ f_1(a) * f_1(V_1 * V_2) \? f_1(b * V_3) * f_1(c * V_4)$         &74     &54     &74   &54    &53      &112       &161  &268        \\
		\hhline{-|-|-|-|-|-|-|-|-|-|}
        $P_{6}$ & $ f_1(V_1) \? f_1(V_2 * V_3 * f_2(V_4,V_5))$                      &21     &1      &1   &1      &4 &17       &1    &1          \\
        $P_{7}$ & $ f_2(V_1,V_2 * V_3 * V_4) \? f_2(V_5 * f_1(V_6 * V_7),V_8)$      &1596   &1      &1   &1      &3473   &41592    &9    &9          \\
        $P_{8}$ & $ f_3(V_1,V_2,V_3) \? f_3(f_1(V_4 * V_5),f_1(V_6 * V_7 * V_8),f_1(f_1(V_9)))$    &399    &1      &1   &1      &507     &3289   &8    &8  \\
        $P_{9}$ & $ f_4(V_1,V_2 * V_3,f_1(V_2 * V_4 * V_5),V_3) \?                                 
                     f_4(f_2(V_6,V_7)* V_6,V_8,V_9,f_1(f_1(V_{10})))$               &492    &14    &1   &1      &122544  &61184      &14   &14         \\
        $P_{10}$ & $ f_5(V_1,V_2 * V_3 * V_4,f_2(V_5,f_1(V_3 * V_4)),                              
                     V_4,f_1(V_6 * V_7)) \? f_5(f_2(V_8,V_{9}),V_{10},V_{11},                     
                     f_1(f_1(V_8)),V_{12})$                                         &161    &11      &1   &1      &6780   &9249   &16   &16      \\
		\hhline{-|-|-|-|-|-|-|-|-|-|}
        $P_{11}$ & $ f_1(V_1 * V_2) \? f_2(V_3 * V_4 * V_5,f_2(V_4,V_5))$                      &0     &0      &0   &0      &985       &125       &1    &1          \\
        $P_{12}$ & $ f_2(V_1,V_2 * V_3 * V_4) \? f_3(V_5 * f_1(V_6 * V_7),V_8,V_9)$      &0     &0      &0   &0      &2987   &57    &1    &1          \\
        $P_{13}$ & $ f_3(V_1,V_2,V_3 * V_4) \? f_2(f_1(V_5 * V_6 * V_7),f_1(f_1(V_8)))$    &0     &0      &0   &0      &468     &48   &1    &1  \\
        $P_{14}$ & $ f_4(V_1,V_2 * V_3,f_1(V_2 * V_4 * V_5),V_3) \?                                 
                     f_3(f_2(V_6,V_7)* V_6,V_8,f_1(f_1(V_{9})))$               &0     &0      &0   &0      &118028  &53653      &1   &1         \\
        $P_{15}$ & $ f_5(V_1,V_2 * V_3 * V_4,f_2(V_5,f_1(V_3 * V_4)),                              
                     V_6,f_1(V_7 * V_8)) \? f_4(f_2(V_9,V_{10}),V_{11},                     
                     f_1(f_1(V_9)),V_{12})$                                         &0     &0      &0   &0      &6968   &7033   &1   &1      \\
        \hhline{-|-|-|-|-|-|-|-|-|-|}
	\end{tabular*}
	}
\caption{Experimental evaluation (exclusive-or)}\label{tab:xor}
\vspace{-2ex}
\end{table}

\rowcolors{1}{gray!10}{white}
\begin{table}[t]
	\centering
	\scriptsize
	{\setlength{\tabcolsep}{0.5em}
	\begin{tabular*}{\textwidth}{|c|>{\centering}p{6.4cm}|r|r|r|r|r|r|r|r|r|}
		\hhline{-|-|-|-|-|-|-|-|-|-|}
		\rowcolor{gray!50}
		\multicolumn{2}{|c|}{\it Unification problem}
		&\multicolumn{1}{c|}{$\#_{\mathit{maude}}$}
		&\multicolumn{1}{c|}{$\#_{\mathit{fast}}$}
		&\multicolumn{1}{c|}{$\#_{\mathit{cr}}$}
		&\multicolumn{1}{c|}{$\#_{\mathit{cr+fast}}$}
		&\multicolumn{1}{c|}{$\mathcal{T}_{\mathit{maude}}$} 
		&\multicolumn{1}{c|}{$\mathcal{T}_{\mathit{fast}}$} 
		&\multicolumn{1}{c|}{$\mathcal{T}_{\mathit{cr}}$} 
		&\multicolumn{1}{c|}{$\mathcal{T}_{\mathit{cr+fast}}$}\\
		\hhline{-|-|-|-|-|-|-|-|-|-|}
        $P_{16}$ & $ V_1 \? V_2 + V_3 + V_4$                                         &3702     &1      &1   &1      &4344602      &5034046     &1    &1          \\
        $P_{17}$ & $ V_1 \? f_3(V_2 + V_3,f_1(V_3 + V_4),f_2(V_2,f_1(V_4)))$         &3789     &1      &1   &1      &6956340     &5413107     &2    &2          \\
        $P_{18}$ & $ V_1 + V_2 \? V_3 + V_4$                                         &3611     &664      &3313   &664     &36258       &547115       &253078   &657746        \\
        $P_{19}$ & $ V_1 + V_2 \? f_2(V_3,f_1(V_4 + V_5))$                     &376    &8      &52   &8      &26425     &5083 &366 &2000       \\
        $P_{20}$ & $ f_1(a) + f_1(V_1 + V_2) \? f_1(b + V_3) + f_1(c + V_4)$         &316     &193     &316   &193    &10202      &4175       &3161  &6976        \\
		\hhline{-|-|-|-|-|-|-|-|-|-|}
        $P_{21}$ & $ f_1(V_1) \? f_1(V_2 + V_3 + f_2(V_4,V_5))$                      &158     &1      &1   &1      &426       &1410       &2    &2          \\
        $P_{22}$ & $ f_2(V_1,V_2 + V_3 + V_4) \? f_2(V_5 + f_1(V_6 + V_7),V_8)$      &-   &-      &1   &1      &T/O   &T/O    &11    &11          \\
        $P_{23}$ & $ f_3(V_1,V_2,V_3) \? f_3(f_1(V_4 + V_5),f_1(V_6 + V_7 + V_8),f_1(f_1(V_9)))$    &-    &-      &1   &1      &T/O     &T/O   &11    &13  \\
        $P_{24}$ & $ f_4(V_1,V_2 + V_3,f_1(V_2 + V_4 + V_5),V_3) \?                                 
                     f_4(f_2(V_6,V_7)* V_6,V_8,V_9,f_1(f_1(V_{10})))$               &-    &-    &1   &1      &T/O  &T/O      &19   &19         \\
        $P_{25}$ & $ f_5(V_1,V_2 + V_3 + V_4,f_2(V_5,f_1(V_3 + V_4)),                              
                     V_4,f_1(V_6 + V_7)) \? f_5(f_2(V_8,V_{9}),V_{10},V_{11},                     
                     f_1(f_1(V_8)),V_{12})$                                         &-    &-      &1   &1      &T/O   &T/O   &24   &24      \\
		\hhline{-|-|-|-|-|-|-|-|-|-|}
        $P_{26}$ & $ f_1(V_1 + V_2) \? f_2(V_3 + V_4 + V_5,f_2(V_4,V_5))$                      &-     &0      &0   &0      &T/O         &5594580       &1    &1          \\
        $P_{27}$ & $ f_2(V_1,V_2 + V_3 + V_4) \? f_3(V_5 + f_1(V_6 + V_7),V_8,V_9)$      &-     &0      &0   &0      &T/O     &4399334    &1    &1          \\
        $P_{28}$ & $ f_3(V_1,V_2,V_3 + V_4) \? f_2(f_1(V_5 + V_6 + V_7),f_1(f_1(V_8)))$    &-     &0      &0   &0      &T/O       &3757585   &1    &1  \\
        $P_{29}$ & $ f_4(V_1,V_2 + V_3,f_1(V_2 + V_4 + V_5),V_3) \?                                 
                     f_3(f_2(V_6,V_7)* V_6,V_8,f_1(f_1(V_{9})))$               &-     &-      &0   &0      &T/O  &T/O      &1   &1         \\
        $P_{30}$ & $ f_5(V_1,V_2 + V_3 + V_4,f_2(V_5,f_1(V_3 + V_4)),                              
                     V_6,f_1(V_7 + V_8)) \? f_4(f_2(V_9,V_{10}),V_{11},                     
                     f_1(f_1(V_9)),V_{12})$                                         &-     &-      &0   &0      &T/O   &T/O   &1   &1      \\
		\hhline{-|-|-|-|-|-|-|-|-|-|}
	\end{tabular*}
	}
\caption{Experimental evaluation (abelian group)}\label{tab:ag}
\vspace{-2ex}
\end{table}

Tables~\ref{tab:xor} and \ref{tab:ag} show that the \textit{cr+fast} combination is the best choice, since it combines the benefits of both the \textit{fast} unification algorithm of \cite{ES19}
and the new constructor-root unification algorithm \textit{cr}. 
For the number of unifiers, \textit{cr} always reported less unifiers than Maude except for problem $P_{20}$, where both report the same number.
However, both the \textit{cr} and the \textit{fast} algorithm are incomparable and \textit{cr} reported less unifiers than \textit{fast} in the unification problems $P_9$ and $P_{10}$, whereas \textit{fast} reported less unifiers than \textit{cr} in the unification problems $P_3, P_5, P_{18}, P_{19}, P_{20}$. 
As for the execution time, \textit{cr} can beat both Maude and the \textit{fast} algorithm for almost all the unification problems.
Indeed, unification in the abelian group is so complex that neither Maude nor \textit{fast} can terminate in most of the unification problems (e.g., $P_{22}, P_{23}, P_{24}, P_{25}$, and more), whereas \textit{cr} did.

Our best contribution are the non-unifiable problems in the third block of Tables~\ref{tab:xor} and \ref{tab:ag}.
Our new constructor-root unification algorithm immediately terminates, whereas neither Maude nor \textit{fast} could, as shown in the unification problems
$P_{11}, P_{12}, P_{13}, P_{14}, P_{15}, P_{26}, P_{27}, P_{28}, P_{29}, P_{30}$.

\section{Conclusion and Future Work}\label{sec:conc}

The variant-based equational unification algorithm implemented in the most recent version of Maude, version 3.0, 
may compute many more unifiers than the necessary or may not be able to stop immediately.
Constructor symbols are extensively used in computer science,
but they have not been integrated into
the variant-based equational unification procedure of Maude.
In this paper, we have redefined the variant-based unification algorithm 
and our experiments on some unification problems show an impressive speedup.
Especially for non-unifiable problems, where many resources are wasted. 

As far as we know, this is the only research line to reduce the number of variant unifiers. The closest work is to combine standard unification algorithms with variant-based unification, such as \cite{EKMN+15,EEMR19}. 
Note that the constructor variant unification of \cite{Meseguer18-scp} is not connected to our work,
since it is based on a new notion of \emph{constructor variant}.
This is just a step forward on developing new techniques for improving variant-based unification and we plan to reduce even more the number of variant unifiers.

\bibliographystyle{eptcs}

\end{document}